\documentclass[aps,pra,longbibliography,superscriptaddress,twocolumn]{revtex4-1}

\usepackage[T1]{fontenc}
\usepackage[utf8]{inputenc}

\usepackage{float}
\usepackage{graphicx}
\usepackage{epsfig,epstopdf}
\usepackage[table]{xcolor}
\usepackage{comment}

\usepackage{mathtools}

\newcommand{\CC}{\mathbb{C}}

\usepackage[english]{babel}

\PassOptionsToPackage{hyphens}{url}
\usepackage[hyphenbreaks]{breakurl}

\usepackage{times}

\usepackage{amsmath}
\usepackage{bbm}
\usepackage{amsfonts}
\usepackage{amssymb}
\usepackage{amsthm}
\usepackage{mathbbol}
\usepackage{mathtools}

\usepackage{float}
\usepackage{graphicx}
\usepackage{epsfig,epstopdf}
\usepackage{xcolor}
\usepackage{tikz}
\usepackage{pgffor}

\usepackage{times}

\usepackage{enumerate}
\usepackage{enumitem}

\usepackage{pgfplots}
\usetikzlibrary{pgfplots.groupplots}
\pgfplotsset{compat=1.11}
\usepgfplotslibrary{fillbetween}

\usepackage{tikz}
\usetikzlibrary{
  shapes,
  shapes.geometric,
	trees,
	matrix,
  positioning,
    pgfplots.groupplots,
  }
\newcommand{\ipic}[3][-0.5]{\raisebox{#1\height}{\scalebox{#3}{\includegraphics{#2}}}}

\PassOptionsToPackage{pdftex,hyperfootnotes=false,pdfpagelabels}{hyperref}
\usepackage{hyperref}
\hypersetup{
    colorlinks=true, linktocpage=true, pdfstartpage=3, pdfstartview=FitV,
    breaklinks=true, pdfpagemode=UseNone, pageanchor=true, pdfpagemode=UseOutlines,
    plainpages=false, bookmarksnumbered, bookmarksopen=true, bookmarksopenlevel=1,
    hypertexnames=true, pdfhighlight=/O,
    urlcolor=blue, linkcolor=blue, citecolor=black,
}


\newtheorem{theorem}{Theorem}

\newtheorem{definition}{Definition}
\newtheorem{lemma}{Lemma}

\DeclareMathOperator{\tr}{Tr}

\definecolor{christian}{rgb}{0,.4,1}

\definecolor{jens}{rgb}{0,0,0}

\newcommand{\new}[1]{{\color{jens}#1}}

\definecolor{roberto}{rgb}{0,.4,1}

\definecolor{juan}{rgb}{.7,.1,0}

\definecolor{dominik}{rgb}{0.4,.0,0.6}

\newcommand{\fu}{Dahlem Center for Complex Quantum Systems, Freie Universit{\"a}t Berlin, 14195 Berlin, Germany}
\newcommand{\tii}{Quantum Research Centre, Technology Innovation Institute (TII), Abu Dhabi}
\newcommand{\hzb}{Helmholtz-Zentrum Berlin f{\"u}r Materialien und Energie, 14109 Berlin, Germany}

\begin{document}

\title{Emergent statistical mechanics
from properties of\\ disordered random matrix product states}

\author{Jonas Haferkamp}
\affiliation{\fu}
\affiliation{\hzb}
\author{Christian Bertoni}
\affiliation{\fu}
\author{Ingo Roth}
\affiliation{\fu}
\affiliation{\tii}
\author{Jens Eisert}
\affiliation{\fu}
\affiliation{\hzb}

\begin{abstract}
The study of generic properties of quantum states has led to an abundance of insightful results. A meaningful set of states that can be efficiently prepared in experiments are ground states of gapped local Hamiltonians, which are well approximated by matrix product states. In this work, we \new{introduce a picture of generic states within the trivial phase of matter with respect to their non-equilibrium and entropic properties: We do so by rigorously exploring non-translation-invariant matrix product states drawn from a local i.i.d.\ Haar-measure.}
 We arrive at these results by 
 \new{exploiting} techniques for computing moments of random unitary matrices and by exploiting a  mapping to partition functions of classical statistical models, a method that has lead to valuable insights on local random quantum circuits.
 Specifically, we prove that such disordered random matrix product states equilibrate exponentially well with overwhelming probability under the time evolution of  Hamiltonians featuring a non-degenerate spectrum. Moreover, we prove two results about the entanglement R\'{e}nyi entropy: The entropy with respect to sufficiently disconnected subsystems is generically extensive in the system-size, and for small connected systems the entropy is almost maximal for sufficiently large bond dimensions.
\end{abstract}

\maketitle

The application of random matrix theory to the study of interacting quantum many-body systems has proven to be a particularly fruitful endeavour, in fact, in various readings. This includes countless applications in condensed matter physics, allowing to predict a wide range of topics ranging from universal conductance fluctuations, weak localization or coherent back-scattering \cite{RMT}. 
In a mindset closer to notions of quantum information theory, it has been shown that for uniformly drawn states, such ideas lead to a principle of maximum entropy~\cite{sanchez1995simple,hayden2004randomizing,hayden2006aspects}, an insight that has implications in the context of quantum computing~\cite{gross2009most}.
 That said, from an operational perspective, quantum states that are uniformly random --- in the sense that they are drawn from a global invariant measure --- are not particularly natural in 
many contexts.
 From a physical perspective, such random states do not respect locality 
present in most naturally occurring systems.
 What is more common and natural, in contrast, are quantum states that emerge from ground states of gapped Hamiltonians.
  Indeed, in common experiments, often good approximations of ground states
of local Hamiltonian characterized by finite-ranged interactions can be feasibly prepared, by means of cooling procedures or by resorting to suitable loading 
procedures. A key question that arises, therefore, is to what extent one can 
expect such ground states to exhibit the same or similar properties as (Haar-random) 
uniformly chosen ones. 

This question can be interpreted in several readings.
A particularly important one \new{from the perspective
of out-of-equilibrium physics} is, specifically, 
to what extent such states would \new{eventually} 
equilibrate in time under the evolution of
general Hamiltonians \cite{1408.5148,PolkovnikovReview,christian_review}.
Equilibration is an important concept in the foundations of quantum statistical 
mechanics and considerations how apparent equilibrium states seem to
emerge under closed system quantum dynamics. 
Equilibration refers to properties
becoming apparently \new{and effectively} stationary, even though the
entire systems would \new{remain to} undergo unitary dynamics \new{generated by some local Hamiltonian}.

In this work, we consider such typical ground states \new{from a fresh
perspective}.
 More precisely, we prove several concentration-type results for so-called  \emph{matrix product states (MPS)}, instances of tensor network states that approximate ground states of
gapped one-dimensional quantum systems well. This class of states hence indeed captures those that can be obtained by cooling local Hamiltonians with a spectral gap.  To pick such random states seems a most meaningful approach to respect natural restrictions
of locality. \new{Since one can readily see such states as being
ground states of disordered parent Hamiltonians \cite{1010.3732}, they can be
viewed as being typical representatives of one-dimensional
quantum phases of matter.}
Indeed, \new{while the significance of random ensembles
of quantum states respecting locality has been appreciated
early on \cite{garnerone2010statistical}, only recently,
first steps have been taken towards a rigorous
understanding of such ensembles, specifically concerning
spectral properties and decays of correlations
of generic MPS
~\cite{gonzales_spectral_2018,Lancien}.}

\new{Specifically, the powerful technical tool we bring in to this kind of study is a framework -- seemingly out of context -- related to mappings to partition functions of a one-dimensional statistical mechanical models. We also develop this picture further. 
Equipped with this technical tool,} we prove that with overwhelming probability, a randomly chosen MPS 
--- drawn according to the i.i.d.\ Haar measure in the physical
and the virtual dimensions of the tensor network state ---
equilibrates exponentially under the time evolution of Hamiltonians with non-degenerate spectrum. Moreover, as a second result, and also motivated by the equilibration of systems exhibiting \emph{many-body localization}
\cite{RevModPhys.91.021001,huse2014phenomenology,1409.1252}, we prove the extensivity of the R\'{e}nyi $2$-entropy with overwhelming probability with respect to equidistant disconnected subsystems.
This result complements recent work obtained for translation-invariant MPS~\cite{rolandi2020extensive}.
A third result is an improved \emph{principle of maximum entropy} 
for disordered random MPS. 
More precisely, we show that the R\'{e}nyi $2$-entropy is almost maximal for small connected subsystems up to errors polynomially small in the bond dimension $D$.
This again complements results for random translation-invariant MPS~\cite{gonzales_spectral_2018,collins2012matrix}.
Another complementary result is~\cite{movassagh2019ergodic,movassagh2020theory}, where formulas for the asymptotic values of various quantum entropies are derived in the setting of ergodic Markov chains.

\new{Again, this substantial progress on studying generic equilibrium
and out-of-equilibrium states of matter are facilitated by a 
technical tool introduced to the context at hand:}
The proofs of the above results are based on a graphical 
calculus, as all of them are obtained via mappings to partition functions of a one-dimensional  statistical mechanical model. These models can be obtained by an application of the \emph{Weingarten calculus}.
In fact, the method we make use of is inspired by recent work on 
\emph{random quantum circuits}. 
Such random quantum circuits have recently prominently
been discussed in the literature, both in the 
context of quantum computing where they are used to show
a quantum advantage over classical algorithms
\cite{neill_blueprint_2017,GoogleSupremacy}, as well as in proxies for
scrambling dynamics 
\cite{chandran2015semiclassical,PhysRevB.98.134204,nahum2018operator,hunter-jones_unitary_2019,zhou2019emergent,ComplexityGrowth,brandao_local_2016,Homeopathy}.
 In that context,
similar mappings have been exploited
\cite{hunter-jones_unitary_2019}.
We showcase \new{here, therefore,} how this mapping can be \new{a very much helpful} 
tool in a conceptually very different context.
We also make use of a generalization of Cauchy-Schwarz inequality to tensor networks without self-contractions~\cite{kliesch_guaranteed_2019}.
In an appendix, we argue that the low-entanglement structure of MPSs prevents them from having generic properties of uniform states in the sense that they will now form an approximate complex projective $2$-design \cite{Designs} in a meaningful sense. 
\new{The upshot of this work is that using a machinery
of mappings to partition functions of a one-dimensional  statistical mechanical model and the Weingarten calculus, a wealth
of out-of-equilibrium and equilibrium properties of 
generic one-dimensional quantum phases of matter can be
rigorously computed.}

\section{Setting}
\subsection{Random matrix product states}

We start by stating the underlying model of random quantum states used throughout this work.
 A very reasonable model of random matrix product states already used in Ref.~\cite{garnerone2010statistical} is the following.
 Consider an arbitrary MPS state vector 
 \cite{quant-ph/0608197,raey} with periodic boundary
conditions and $n$ constituents of local dimension $d$.
Such a state vector can be written as
\begin{equation}
|\psi\rangle =\sum_{i_1,\dots, i_n} \tr\left[A_{i_1}^{(1)} A_{i_2}^{(2)}\dots A_{i_n}^{(n)}\right]|i_1,\dots, i_n\rangle.
\end{equation}
Here, $A^{(i)}$ are \new{complex valued} 
$D\times D$ matrices \new{that specify
the quantum state at hand.}
$D$ is referred to as the \emph{bond dimension},
\new{sometimes also called the tensor train
rank}. In a commonly used graphical calculus, this is represented as
\begin{equation}
\ipic{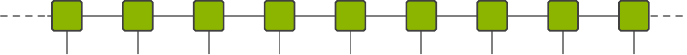}{0.7}.
\end{equation}
\new{We consider MPSs that can be unitarily embedded} \cite{MPSSurvey,Webs}: Here, one equips each tensor
with another leg and feeds in an arbitrary state vector $|0\rangle\in \CC^d$, to get
\begin{equation}\label{eq:stateUprep}
	\ipic{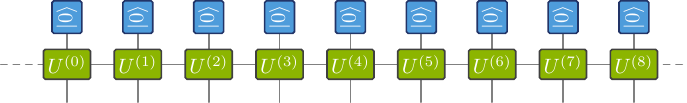}{0.7} .
\end{equation}
Each unitary $U^{(0)}, \dots , U^{(n)}\in U(d D)$ can be seen as mapping one $\CC^d\otimes \CC^D$  input system to another  $\CC^d\otimes \CC^D$ 
output system. 
These can be normalized by choosing an appropriate normalization for the boundary vectors. 
We can construct MPS with periodic boundary conditions analogously. 
This motivates a very natural probability measure.

\begin{definition}[Random matrix product state] A random matrix product state (RMPS) of local dimension $d$, system size $n$ and bond dimension $D$ is a state defined by \eqref{eq:stateUprep} with each unitary 
	$U^{(1)},\dots, U^{(n)} \in U(dD)$ drawn i.i.d.\ randomly from the Haar measure. 
	We denote the resulting measure as $\mu_{d,n,D}$. 
\end{definition}
Note that this definition can be regarded as drawing the tensor cores $A^{(i)}$s uniformly from the Stiefel manifold of isometries.  
This probability measure makes a lot of sense: It is a distribution over random disordered, non-translation-invariant quantum states. \new{Once again, each} 
realization will have a \emph{parent Hamiltonian} \cite{1010.3732}, a gapped local Hamiltonian for which the given state is an exact ground state. In this sense, the probability measure discussed here can equally be seen as a probability measure on disordered, random local Hamiltonians. 
It should be noted that this 
probability measure takes values on state vectors that are not exactly normalized. 
It is, however, easy to see that the values of $\langle \psi|\psi\rangle$ are strongly concentrated around unity.
\begin{lemma}[Concentration 
around a unit norm]\label{lemma:boundonnorm}
 \begin{equation}
 \mathrm{Pr}\left(|\langle\psi|\psi\rangle-1|\geq \varepsilon\right)\leq \varepsilon^{-2} d^{-n}.
 \end{equation}
\end{lemma}
This will be proven later, in Section~\ref{sec:normconcentration}.

\subsection{Effective dimension and equilibration}

We now turn to discussing concepts of equilibration in quantum many-body dynamics. Equilibration refers to the observation that for an 
observable $A$, the expectation value of a time-evolving quantum many-body system will for the overwhelming times take the same values as the expectation value with respect to the infinite time average: The expectation value then looks ``equilibrated'' 
\cite{Linden_etal09,christian_review}.
For a given arbitrary 
state vector $|\psi\rangle$ that reflects the initial 
pure state at time $T=0$,
this expectation value of an observable $A$ with respect to the infinite time average takes the form
\begin{equation}
A^{\infty}_{\psi}:=\lim_{t\to\infty}\frac{1}{t}\int_{0}^{t}\langle\psi|A(t')|\psi\rangle\mathrm{d}t'.
\end{equation}
The fluctuations around this infinite time average are defined as
\cite{Linden_etal09}
\begin{equation} \label{eq:itavg}
\Delta A^{\infty}_{\psi}:=\lim_{t\to\infty}\frac{1}{t}\int_{0}^t|\langle\psi|A(t')|\psi\rangle-A^{\infty}_{\psi}|^2\mathrm{d}t'.
\end{equation}
The reduced density matrices on a local region are determined by local observables.
It hence suffices to consider the above definitions.
In particular, if $\Delta A^{\infty}_{\psi}\ll 1$ then the reduced density matrices on a small region can not deviate from the infinite time average except for very brief time periods. 
In this sense a small value for $\Delta A^{\infty}_{\psi}$ implies equilibration.

\section{Equilibration of random matrix product states}

Equipped with these preparations, we are in the position to state
our first main result. On the equilibration of RMPS, 
\new{following non-equilibrium dynamics,} we will 
prove the following theorem.

\begin{theorem}[Equilibration of  RMPS]\label{theorem:equilibration}
	Let $H$ be a Hamiltonian with non-degenerate spectrum and non-degenerate spectral gap and $A$ an observable evolving in time governed by $H$. 
	Let $|\psi\rangle$ be a RMPS drawn from $\mu_{d, n, D}$ and $|\psi'\rangle:={|\psi\rangle}/{\sqrt{\langle \psi|\psi\rangle}}$.
	Then there are constants \new{independent of the system parameters} $c_1,c_2$ such that for $n$ sufficiently large the infinite time average $\Delta A^{\infty}_{\psi'}$, \eqref{eq:itavg}, fulfills 
	\begin{equation}
	\mathrm{Pr}
\left(\Delta A^{\infty}_{\psi'}\leq e^{-c_1\alpha(d,D)n)}\right)
	\geq 1-e^{-c_2\alpha(d,D)n}
	\end{equation}
	with
	\begin{equation}\label{eq:alphadD}
	\alpha(d,D)=\log\left(\frac{d-\frac{1}{dD^2}}{(1+\frac{1}{D})(1+\frac{1}{dD})}\right) \,.
	\end{equation}

\end{theorem}
This shows that under the time evolution of many Hamiltonians almost all matrix product states equilibrate exponentially well.
The proof of Theorem~\ref{theorem:equilibration} implies a second result which we can informally state as follows:
	Assume we draw Hamiltonians from an ensemble such that all marginal distributions for most eigenstates are distributed according to the RMPS measure introduced above. 
	Then, for an arbitrary initial state, the system equilibrates exponentially well with overwhelming probability. 
This is motivated by the fact that the equilibration of systems exhibiting many-body localization still lacks a complete rigorous explanation. In particular, it is known that the energy eigenstates of many-body localized systems satisfy an \emph{area law for the entanglement entropy}~\cite{Bauer,AreaReview,1409.1252} and can be described by matrix product states~\cite{1409.1252}.
It seems unrealistic that any natural ensemble of Hamiltonians has RMPS as marginal eigenstates but we believe that enough of this result might survive for more structured ensembles to prove equilibration of many-body localized systems.

\subsection{Proof techniques}

In this section, we present many of the proof
techniques relevant for this work. In order to
obtain Theorem~\ref{theorem:equilibration}, we observe that the result in Ref.~\cite[Thm.~3]{huang_instability_2019} can be generalized to any distribution $\nu$ on states provided that the effective dimension is large.
The effective dimension measures how much overlap the initial state has with the energy eigenstates of the Hamiltonian. 
In particular, we have the following key result from Refs.\ \cite{huang_instability_2019,tasaki_quantum_1998,reimann_foundation_2008} that we make use of.

\begin{lemma}[\cite{huang_instability_2019,tasaki_quantum_1998,reimann_foundation_2008}]
\label{lemma:effectivedimension}
Consider a Hamiltonian with non-degenerate spectrum and non-degenerate spectral gaps, i.e. $E_n-E_m= E_j-E_k$ if and only if $n= j$, $m=k$ where $E_n$ labels the eigenvalues of the Hamiltonian.
Then,
\begin{equation}
\Delta A^{\infty}_{\psi}=O(1/D_{\rm eff}),
\end{equation}
with 
\begin{equation}\label{eq:effectivedimension}
1/D_{\rm eff}:=\sum_j |\langle \psi|j\rangle|^4,
\end{equation}
where $\{|j\rangle\}$ is the eigenbasis of the Hamiltonian $H$.
\end{lemma}
Theorem~\ref{theorem:equilibration} follow immediately from Lemma~\ref{lemma:effectivedimension} in combination with the following 
statement.

	\begin{lemma}[Bound to effective dimension]\label{lemma:overlap}
		For all state vectors $|\phi\rangle$ we have
	\begin{equation}\label{eq:overlapbound}
\mathbb{E}_{\psi\sim \mu_{d,n,D}}|\langle \psi|\phi\rangle|^4\leq 2\frac{(1+\frac{1}{D})^{n}(1+\frac{1}{dD})^{n}}{(d^2-\frac{1}{D^2})^n}\,.
	\end{equation}
\end{lemma}
We are trying to find an upper bound on 
\begin{equation}\label{eq:effectivedim}
\mathbb{E}|\langle \psi|\phi\rangle|^4
= \langle \phi|^{\otimes 2}\mathbb{E}(|\psi\rangle\langle \psi|)^{\otimes 2}|\phi\rangle^{\otimes 2}.
\end{equation}
We make use of the Weingarten calculus~\cite{collins2006integration,brouwer1996diagrammatic}, as elaborated upon here in the following statement.

\begin{lemma}[Weingarten calculus]\label{theorm:weingarten} 
The $t$-th moment operator of Haar-random unitaries is given by
\begin{equation}\label{eq:weingarten}
\mathbb{E}_{U\sim\mu_H} U^{\otimes t}\otimes \overline{U}^{\otimes t} = \sum_{\sigma,\pi\in S_t} \mathrm{Wg}(\sigma^{-1}\pi,q)|\sigma\rangle\langle \pi|,
\end{equation}
where $|\sigma\rangle:=(\mathbb{1}\otimes r(\sigma))|\Omega\rangle$ with $r$ being the representation of the symmetric group  $S_t$ on $(\mathbb C^{q})^{\otimes t}$ which permutes the vectors in the tensor product, and $|\Omega\rangle=\sum_{j=1}^{q^t}|j,j\rangle$ the maximally entangled state vector up to normalization.
\end{lemma}
The Weingarten calculus is a powerful tool in particular when suitably combined with Penrose tensor-network diagrams providing a graphical calculus:
E.g., for $t=2$ and $U(q)$ the formula~Eq.~\eqref{eq:weingarten} takes the form
\begin{multline}
\mathbb{E}_{U\sim\mu_H}\ipic{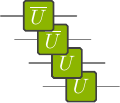}{0.8}\\=\frac{1}{q^2-1}\left[\;\;\ipic{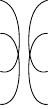}{0.8}-\frac{1}{q}\ipic{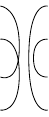}{0.8}-\frac{1}{q}\ipic{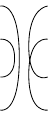}{0.8}+\ipic{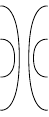}{0.8}\;\;\right]\;.\label{eq:graphweingarten}
\end{multline}
	Graphically, we can express~\eqref{eq:effectivedim} as
	\begin{figure}[H]
	\begin{multline}
\mathbb{E}|\langle \psi|\phi\rangle|^4
=\mathbb{E}_{U^{(i)}\sim \mu_H}\\\ipic{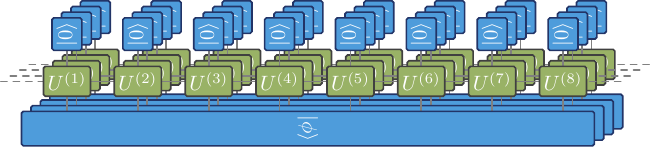}{0.8}\,.
	\end{multline}
	\end{figure}
	By evaluating each $\mathbb{E}_{U^{(i)}}U^{(i)}\otimes U^{(i)}\otimes \overline{U}^{(i)}\otimes \overline{U}^{(i)}$ individually according to~\eqref{eq:graphweingarten}, this can be reformulated as a partition function. 
	Introducing the notation 
	\begin{equation}
	|\psi\rangle^{\otimes 2,2}:=|\psi\rangle^{\otimes 2}\otimes \overline{|\psi\rangle}^{\otimes 2}, 
	\end{equation}
	we obtain
	\begin{multline}\label{eq:partitionfunction}
\mathbb{E}|\langle \psi|\phi\rangle|^4\\
=\sum_{\{\mathbb{1},\mathbb{F}\}^{2n}}\ipic{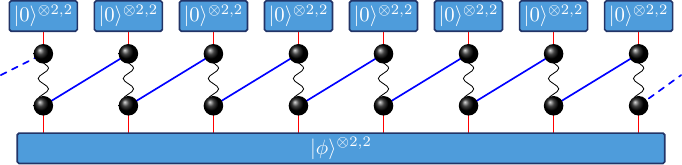}{0.6}.
	\end{multline}
	 Here, the black balls correspond to a choice of an element of
	 $S_2=\{\mathbb{1},\mathbb{F}\}$ with $\mathbb{F}$ the swap permutation.
	 The wiggly line corresponds to different weights for every pair of permutations $(\pi,\sigma)$ with the corresponding value of the Weingarten function $\mathrm{Wg}(\pi^{-1}\sigma,q)$ according to~\eqref{eq:graphweingarten} with $q=dD$.
	 The red edges denote contractions over $\CC^d$ and the blue edges are contractions over $\CC^D$.
	 This is reminiscent of Ref.~\cite{hunter-jones_unitary_2019}, where the frame potential of random quantum circuits is mapped to a partition function with local degrees of freedom corresponding to permutations.
	 Notice, however, that for every permutation $\pi\in S_2$, we always have \begin{equation}
	     \langle \pi|0\rangle^{\otimes 2,2}=|\langle 0|0\rangle|^2=1.
	 \end{equation}
	Moreover, every summand contains a factor of the form $\langle \phi|^{\otimes 2,2}\bigotimes_{l=1}^n|\sigma_{l}\rangle$.
	We can bound this contribution using the following generalization of the Cauchy-Schwarz inequality~\cite{kliesch_guaranteed_2019}.
	\begin{lemma}[Cauchy-Schwarz inequality
	for tensor networks \cite{kliesch_guaranteed_2019}]
		Consider a tensor network $(T,C)$ with $J\geq 2$ tensors $T=(t^j)_{\in \{1,\cdots, J\}}$ such that no tensor in the contraction $C$ self-contracts, i.e. no string connects a tensor with itself.
		Then,
		\begin{equation}
		|C(T)|\leq \prod_{j=1}^J \left|\left|t^j\right|\right|_F,
		\end{equation}
		where $||.||_F$ is the Frobenius norm of the tensor $t^j$ viewed
		as a vector.
	\end{lemma}
\begin{proof}[Proof of Lemma~\ref{lemma:overlap}]
As the tensor network contraction $\langle \phi|^{\otimes 2,2}\bigotimes_{l=1}^n|\sigma_{l}\rangle$ does not contain self-contractions, this yields 
\begin{equation}
|\langle \phi|^{\otimes 2,2}\bigotimes_{l=1}^n|\sigma_{l}\rangle|\leq |||\phi\rangle||_F^4=1.
\end{equation}
Therefore, we can apply a triangle inequality to the sum in~\eqref{eq:partitionfunction} to obtain the  bound
\begin{multline}\label{eq:1/DBound}
 \mathbb{E}|\langle \psi|\phi\rangle|^4\leq\\ \sum_{\{\mathbb{1},\mathbb{F}\}^{2n}}\left|\ipic{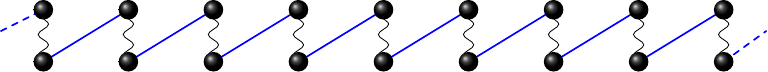}{0.55}\right|.
\end{multline}
As contractions over $\mathbb{C}^D$ we have $\langle\mathbb{1}|\mathbb{1}\rangle=\langle\mathbb{F}|\mathbb{F}\rangle=D^2$ and 
\begin{equation}
\langle\mathbb{1}|\mathbb{F}\rangle=\langle\mathbb{F}|\mathbb{1}\rangle=D.
\end{equation}
This allows us to obtain a sufficient upper bound on $1/D_{\rm eff}$ via a combinatorial argument. 
\new{Consider a sequence $(\sigma_1\pi_1),(\sigma_2\pi_2),\dots, (\sigma_n\pi_n)\in \{\mathbb{1},\mathbb{F}\}^{2n}$. Here, the $\sigma$ refer to the balls on top in equation \eqref{eq:1/DBound} and $\pi$ to the ones on the bottom. If $\sigma_i=\pi_i$, the total interaction between sites $i$ and $i+1$ contributes with a term
\begin{equation}
    \frac{D^2}{d^2D^2-1}
\end{equation}
If $\sigma_i\neq \pi_i$, this contribution is divided by $dD$, and if $\pi_i\neq \sigma_{i+1}$ it is divided by $D$. Hence we sum over all possible ways of choosing $\sigma_i\neq \pi_i$ or $\pi_i\neq \sigma_{i+1}$, and divide by the corresponding factor, to get}
\begin{align}
\begin{split}
\mathbb{E}|\langle\phi|\psi\rangle|^4&\leq 2\frac{D^{2n}}{(d^2D^2-1)^n}\sum_{l=0}^{n}\sum_{r=0}^n{n\choose l}{n\choose r}D^{-l-r}d^{-r}\\
&=2\frac{(1+\frac{1}{D})^{n}(1+\frac{1}{dD})^{n}}{(d^2-\frac{1}{D^2})^n}.
\end{split}
\end{align}
This implies Lemma~\ref{lemma:overlap}.
\end{proof}

Theorem~\ref{theorem:equilibration} now follows from applying Lemma~\ref{lemma:overlap} to~\eqref{eq:effectivedimension} together with an application of Markov's inequality and Lemma \ref{lemma:boundonnorm}.
\begin{proof}[Proof of Theorem~\ref{theorem:equilibration}]
By Lemma \ref{lemma:overlap} we have \begin{equation}
\begin{aligned}
\mathbb{E}(1/D_{\mathrm{eff}})&=\sum_{j}\mathbb{E}(|\langle\psi|j\rangle|^4)\\&\leq 2 \frac{\left(1+\frac{1}{D}\right)^n\left(1+\frac{1}{dD}\right)^n}{\left(d-\frac{1}{dD^2}\right)^n}\\
&=2 e^{-\alpha n} 
\end{aligned}
\end{equation}
with $\alpha = \alpha(d, D)$ as defined in \eqref{eq:alphadD}.
 Picking two positive constants $k_1,
 k_2>0$ with $k_1<\frac 12$, since $\Delta A_\psi^\infty=O(1/D_{\mathrm{eff}})$, Markov's inequality yields
	\begin{equation}
	\mathrm{Pr}
\left(\Delta A^{\infty}_{\psi}\leq e^{-k_1\alpha n}\right)
	\geq 1-e^{-k_2\alpha n}.
	\end{equation}
	Let $N:=\langle\psi|\psi\rangle$ and let $|\psi'\rangle=N^{-1/2}|\psi\rangle$ be the normalized state vector. We then have 
	\begin{equation}
	\Delta A_{\psi'}^\infty=\frac{\Delta A_{\psi}^\infty}{N^2}.
	\end{equation}
	Suppose $\Delta A^{\infty}_{\psi}\leq e^{-k_1\alpha n}$ and $|N-1|\leq e^{-k_1\alpha n}$. Then $N^2\geq(1-e^{-k_1\alpha n})^2 $ and
	\begin{equation} 
	    \Delta A_{\psi'}^\infty \leq \frac{e^{-k_1\alpha n}}{(1-e^{-k_1\alpha n})^2}\leq e^{-c_1\alpha n}
	\end{equation}
	for some constant $c_1>0$ and the integer $n$ 
	being large enough. Then, by the union bound, we get
	\begin{equation}
	\begin{aligned}
&\mathrm{Pr}
\left(\Delta A^{\infty}_{\psi'}\leq e^{-c_1\alpha n}\right)\\
&\geq1- \mathrm{Pr}\left(\Delta A^{\infty}_{\psi}\geq e^{-k_1\alpha n}\textrm{ or }|N-1|\geq e^{-k_1\alpha n} \right)\\
&\geq \mathrm{Pr}\left(\Delta A^{\infty}_{\psi}\leq e^{-k_1\alpha n}\right)-\mathrm{Pr}\left(|N-1|\geq e^{-k_1\alpha n} \right)\\
&\geq1-e^{-k_2\alpha n}-\mathrm{Pr}\left(|N-1|\geq e^{-k_1\alpha n}\right) \\&\geq1-e^{-k_2\alpha n}-e^{-2k_1\alpha_n}d^{-n}\\&\geq 1-e^{-k_2\alpha n}-e^{-(1-2k_1)\alpha n}
	\end{aligned}
	\end{equation}
	where we have used Lemma \ref{lemma:boundonnorm} and $\alpha\leq\log d$. Finally,  
		\begin{equation}
	\begin{aligned}
\mathrm{Pr}
\left(\Delta A^{\infty}_{\psi'}\leq e^{-c_1\alpha n}\right)&\geq1-e^{-k_2\alpha n}-e^{-\alpha(1-2k_1)n} \\&\geq1-e^{-c_2\alpha n}
	\end{aligned}
	\end{equation}
	for a constant $c_2>0$ and $n$ sufficiently
	large.
\end{proof}

\section{Extensivity of
the $2$-R\'{e}nyi entanglement entropy} 

In this section, we turn to proving our second result.
This is once again motivated by the phenomenon of equilibration 
in quantum many-body physics and the endeavour to provide a rigorous foundation for it. 
It has been proven in Ref.~\cite{wilming_entanglement_2019} that systems equilibrate if their energy eigenstates have R\'{e}nyi entropy that is extensive in the system size $n$, which means that
\begin{equation}\label{eq:extensiveentropy}
S_2\left(\tr_A[|j\rangle\langle j|]\right)\geq g(j)n
\end{equation}
for a sufficiently well-behaved function $g$ and \new{\textit{some} subsystem $A$}.
 This property has been dubbed ~\textit{entanglement ergodicity} \cite{wilming_entanglement_2019}.
Motivated by this insight, it has been proven in Ref.~\cite{rolandi2020extensive} that generic translation-invariant MPS have extensive R\'{e}nyi entropy if one considers a bi-partition \new{of the chain into the subsystem that corresponds to every $k$th site and the rest.}
That is, the entropy grows proportional to 
the boundary $|\partial A|$.
\new{Considering this particular partitioning is enough, since one partition of the system satisfying \eqref{eq:extensiveentropy} suffices to consider the system entanglement ergodic, and to show the equilibration properties proven in \cite{wilming_entanglement_2019}.}
Here, we prove such a result with explicit quantitative bounds for disordered RMPS with overwhelming probability. \new{Interestingly, the details of the correlation length of the state is of no concern for the bound presented.
Even though the correlation length is expected to be small compared to the system size but non-zero
(compare also the results of
Ref.\ \cite{Lancien} in the translationally invariant case), 
we still arrive at an extensive bound of the respective entropy, independent of how the distance $k$ between the sites precisely relates to the correlation length.}

\begin{theorem}[Extensivity of entanglement entropies]\label{theorem:extensiverenyi}
	Suppose $n$ is divisible by \new{a positive integer} $k$ and $A$ consists of \new{what remains after tracing out} every $k$-th qudit. Let $\rho_A'$ be the normalized density matrix of a RMPS drawn from $\mu_{d, n, D}$ reduced on $A$.
	Then,
	\begin{equation}
	\mathrm{Pr}\left(S_2(\rho_A')\geq \Omega\left(\frac{n}{k}\right)\right)\geq 1-e^{-\Omega(n/k)}.
	\end{equation}
\end{theorem}

\begin{proof}
The proof of Theorem~\ref{theorem:extensiverenyi} follows similar lines as the one of Lemma~\ref{lemma:overlap}.
We first show that the purity (\new{$\tr[\rho_A^2]$}) is almost minimal and then apply Markov's inequality.
We have \new{for the expected 
purity of a subsystem}
\new{\begin{align}
\begin{split}
\tr[\rho_A^2]&=\mathbb{E}_{\psi\sim\nu}\left(\tr\left[\rho_A^2\right]\right)\\
&=\mathbb{E}_{\psi\sim\nu}\left(\tr\left[\mathbb{F}_{A,\overline{A}}\rho^{\otimes 2}_A\right]\right)\\
&=\tr\left[\mathbb{F}_{A,\overline{A}}\tr_{B,\overline{B}}\left[\mathbb{E} \left(|\psi\rangle\langle\psi|\right)^{\otimes 2}\right]\right]\\
&=\tr\left[\mathbb{F}_{A,\overline{A}}\otimes \mathbb{1}_{B,\overline{B}}\mathbb{E} \left(|\psi\rangle\langle\psi|\right)^{\otimes 2}\right].
\end{split}
\end{align}}
In the same graphical notation as in the proof of Lemma~\ref{lemma:overlap}, this amounts to
\begin{multline}\label{eq:partitionrenyi}
\mathbb{E}\tr[\rho_A^2]=\sum_{\{\mathbb{1},\mathbb{F}\}^{2n}}\ipic{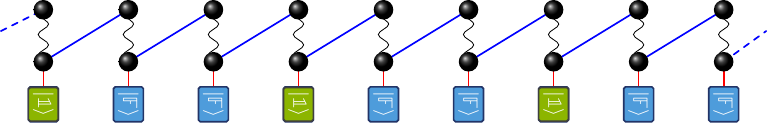}{0.42},
\end{multline}
here depicted for the subset $A$ corresponding to every $k=3$ spin.
Similar to the strategy laid out in Ref.~\cite{hunter-jones_unitary_2019}, we sum over every lower ball to obtain a statistical model. 
For this, we define the following two interactions. These are
\begin{equation}
\ipic[-0.2]{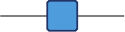}{0.7} :=\sum_{\mathbb{1},\mathbb{F}}\; \ipic{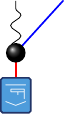}{0.7},\qquad \ipic[-0.2]{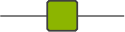}{0.7} :=\sum_{\mathbb{1},\mathbb{F}}\; \ipic{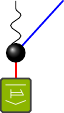}{0.7}.
\end{equation}
Graphically, this yields the following chain
\begin{multline}\label{eq:renyichain}
\mathbb{E}\tr[\rho_A^2]\\
=\sum_{\{\mathbb{1},\mathbb{F}\}^{n}} \ipic{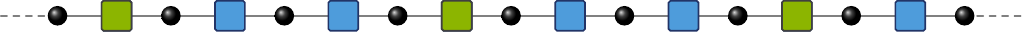}{0.4}\,.
\end{multline}
\new{In order to proceed,} we compute
\begin{align}
\begin{split}\label{eq:plaquettes}
&\mathbb{F}\ipic[0]{plaquette1}{0.7} \mathbb{F}=1,\\ &\mathbb{1}\ipic[0]{plaquette1}{0.7} \mathbb{1}=\frac{dD^2-d}{D^2d^2-1}=\eta(d,D),\\
&\mathbb{1}\ipic[0]{plaquette1}{0.7} \mathbb{F}=0,\\ &\mathbb{F}\ipic[0]{plaquette1}{0.7} \mathbb{1}=\frac{d^2D-D}{D^2d^2-1}=\eta(D,d)
\end{split}
\end{align}
with the notation
\begin{equation}
\eta(x,y):=\frac{xy^2-x}{x^2y^2-1}.
\end{equation}
For the green plaquettes, we simply switch the roles of $\mathbb{1}$ and $\mathbb{F}$.
\new{As an intermediate step, let us compute the same but with $j$ blue plaquettes. The sum over the black balls is implicit, and the following is obtained by simply counting how many switches between $\mathbb 1$ and $\mathbb F$ are allowed.
\begin{align}
\begin{split}\label{eq:jblueplaquettes}
&\mathbb{F}\;\ipic[-0.3]{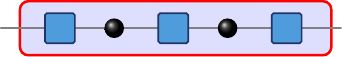}{0.7}\; \mathbb{F}=1,\\ &\mathbb{1}\;\ipic[-0.3]{renyiblocksblue}{0.7}\; \mathbb{1}=\eta(d,D)^j,\\
&\mathbb{1}\;\ipic[-0.3]{renyiblocksblue}{0.7}\; \mathbb{F}=0,\\ &\mathbb{F}\;\ipic[-0.3]{renyiblocksblue}{0.7}\; \mathbb{1}=\eta(D,d)\sum_{l=1}^j \eta(d,D)^{j-l}\\
&\textcolor{white}{aaaaa}=\eta(D,d)\frac{1-\eta(d,D)^j}{1-\eta(d,D)}
\end{split}
\end{align}
}

We can group the chain in~Eq.\ \eqref{eq:renyichain} into blocks of $\new{k-1}$ spins consisting of all spins \new{to the right of} a green plaquette except the one before the next \new{green plaquette}.
For $k\geq 2$ we have, \new{using equation \eqref{eq:jblueplaquettes} with $j=k-1$ combined with \eqref{eq:plaquettes} for the final green plaquette}
\begin{align}
\begin{split}\label{eq:renyiblocks}
&\mathbb{1}\;\ipic[-0.3]{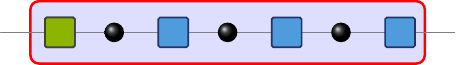}{0.5}\;\mathbb{1}\\
&\textcolor{white}{aaaaa}= 
\eta(d,D)^{\new{k-1}}+\eta(D,d)^2\frac{1-\eta(d,D)^{k-1}}{1-\eta(d,D)}\\
&\textcolor{white}{aaaaa}\leq \eta(d,D)^{\new{k-1}}+\eta(D,d),\\
&\mathbb{1}\;\ipic[-0.3]{renyiblocks}{0.5}\;\mathbb{F}\\
&\textcolor{white}{aaaaa}=\new{\eta(D,d)}\\
&\textcolor{white}{aaaaa}\leq \eta(d,D)^{\new{k-1}}+\eta(D,d),\\
&\mathbb{F}\;\ipic[-0.3]{renyiblocks}{0.5}\;\mathbb{1}\\
&\textcolor{white}{aaaaa}=\eta(d,D)\eta(D,d)\frac{1-\eta(d,D)^{k-1}}{1-\eta(d,D)}\\
&\textcolor{white}{aaaaa}\leq \eta(D,d)+\eta(d,D)^{k-1},\\
&\mathbb{F}\;\ipic[-0.3]{renyiblocks}{0.5}\;\mathbb{F}
=\eta(d,D).
\end{split}
\end{align}
\new{In this step, we have 
used that $\eta(x,y)\leq {1}/{x}$ and for $x\geq 2$
\begin{equation}\label{eq:etaddbound}
\frac{1}{1-\eta(x,y)}\leq 2.
\end{equation}}
We sum over the number of spins with value $\mathbb{1}$.
We make use of the fact that every spin with value $\mathbb{1}$ contributes a factor of $ \eta(d,D)^{\new{k-1}}+\eta(D,d)$ (via the plaquette to its right),
to find
\begin{align}
\begin{split}\label{eq:boundextensivity}
\mathbb{E}\tr[\rho_A^2]&\leq \sum_{i=0}^{\frac{n}{k}}{\frac{n}{k}\choose i}\left(\eta(d,D)^{\new{k-1}}+\eta(D,d)\right)^i\eta(d,D)^{\frac{n}{k}-i}\\
&=\left(\eta(d,D)^{\new{k-1}}+\eta(D,d)+\eta(d,D)\right)^{\frac{n}{k}}\\
&= \eta(d,D)^{\frac{n}{k}}\left(1+\eta(d,D)^{\new{k-2}}+\eta(D,d)\eta(d,D)^{-1}\right)^{\frac{n}{k}}.
\end{split}
\end{align}
Hence, for $d^{\frac12}\leq D$ and $k$ large enough, this expectation value becomes arbitrarily close to the minimal value $d^{-{n}/{k}}$.
 In the regime $D^2\leq d$, the $2$-Renyi entropy is bounded by the 
 entanglement area law \cite{AreaReview}
 \begin{equation}
 |\partial A|\log(D)=\frac{2n}{k}\log(D).
 \end{equation}
 As in the proof of Theorem~\ref{theorem:equilibration}, the expression for the expectation value in Eq.~(\ref{eq:boundextensivity}) 
 implies a concentration result via Markov's inequality and the union bound. 
 This concentration inequality extends to the R\'{e}nyi $2$-entropy. 
\end{proof}

\section{Maximum entropy for small connected subsystems}\label{sec:maximumentropy}

In this section, we will show that small connected
subsystems will with high probability feature
a close to maximum entropy.
A principle of maximum entropy for translation-invariant RMPS has proven in Refs.~\cite{collins2012matrix,gonzales_spectral_2018}.

\begin{theorem}[Almost maximum entropy for reduced states]
Let $A$ be a subset of $l$ of consecutive qudits, let $\rho_A'$ be the normalized density matrix of a RMPS drawn from $\mu_{d, n, D}$ reduced to $A$. Then for any real $r$
\begin{equation}
    \mathrm{Pr}\left(\tr[\rho_A'^2]\geq \Omega(D^{-r})+d^{-l}\right)\leq O\left(D^{-(2-r)}\right)
\end{equation}
for \begin{equation}
n\geq 2\frac{\log D}{\log d}+l.
\end{equation}
\end{theorem}
\begin{proof}
Let $\rho_A$ be the unnormalized 
quantum state. In the disordered case we consider here, we obtain analogously to the previous section
the expression
\begin{multline}
	\mathbb{E} \tr[\rho_A^2]\\
	=\sum_{\{\mathbb{1},\mathbb{F}\}^n}\ipic{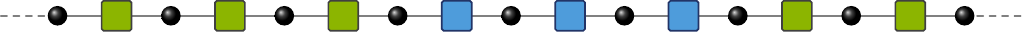}{0.4}
\end{multline}
for the subset $A$ consisting of $l$ consecutive qudits.
Summing over the four contributions \new{in~\eqref{eq:jblueplaquettes}, we obtain}
\new{\begin{align}\label{eq:entropyexpectationbound}
\begin{split}
\mathbb{E} \tr[\rho_A^2] &= \eta(d,D)^{n-l}\\
&+\eta(d,D)^l\\
&+0\\
&+\eta(D,d)^2\frac{1-\eta(D,d)^{\new{l}}}{1-\eta(D,d)}\frac{1-\eta(D,d)^{n-l}}{1-\eta(D,d)}.\\
\end{split}
\end{align}}
By equation \eqref{eq:etaddbound} for $d,D\geq 2$,
\new{\begin{equation}
    \mathbb{E} \tr[\rho_A^2]\leq \frac{1}{d^l}+\frac{1}{d^{n-l}}+4\frac{1}{D^2}
\end{equation}}
holds true.
Let $N=\tr[\rho]$ be the norm squared of the MPS. We have 
\begin{equation}
\tr[\rho_A^2]\geq N^2\frac{1}{d^l} 
\end{equation}
and using (see Eq.\  \eqref{eq:expectedvaluenorm} in the proof of Lemma \ref{lemma:boundonnorm})
\begin{equation}
    \mathbb E(N^2)=1+\eta(d,D)^n\geq 1
\end{equation}
we find
\begin{equation}
\begin{aligned}
    \mathbb E(\tr[\rho_A^2]- N^2d^{-l})&\leq  \new{4}\frac{1}{D^2}+\frac{1}{d^{n-l}}+\new{\frac{1}{d^{n+l}}}\\&\leq \new{6}\frac{1}{D^2},
    \end{aligned}
\end{equation}
where we have used 
\begin{equation}
n\geq 2\frac{\log D}{\log d}+l,\,\text{i.e.}\;d^{l-n}\leq \frac{1}{D^2}.
\end{equation}
Markov's inequality then yields
\begin{equation}
    \mathrm{Pr}( \tr[\rho_A^2]-N^2d^{-l}\geq D^{-r})\leq \new{\frac16} D^{-(2-r)}.
\end{equation}
We now need to ensure that the normalization of the state does not worsen the bound. We employ the union bound in a similar manner as in the proof of Theorem \ref{theorem:equilibration}. Let $\rho_A'={\rho_A}/{N}$ be the normalized state. Suppose that 
\begin{equation}    
    \tr[\rho_A^2]-N^2d^{-l}\leq D^{-r}
\end{equation}
and $|N-1|\leq \epsilon$, 
then $N^2\geq (1-\epsilon)^2$ and
\begin{equation}
    \tr[\rho_A'^2]-d^{-l}\leq \frac{D^{-r}}{
    (1-\epsilon)^2}.
\end{equation}
Pick $\epsilon:=td^{-kn}$, with $k\leq 1-r/4$ and $t<1$. Then $\epsilon \leq t\,d^{kl}D^{-2k}$ and 
\begin{equation}
     \tr[\rho_A'^2]-d^{-l}\leq 
     O(D^{-r}).
\end{equation}
By the same union bound argument as in the proof of Theorem \ref{theorem:equilibration},
we get
\begin{equation}
\begin{aligned}
    &\mathrm{Pr}( \tr[\rho_A']-d^{-l}\leq O(D^{-r}))\\&\geq  \mathrm{Pr}( \tr[\rho_A]-N^2d^{-l}\leq O(D^{-r}))- \mathrm{Pr}(|N-1|\leq d^{-kn}) \\&\geq 1- O(D^{-(2-r)})-d^{-(1-2k)n}\geq 1- O\left(D^{-(2-r)}\right),
    \end{aligned}
\end{equation}
where we have used
that 
$k\leq 1-r/4$ and hence $d^{(1-2k)n}\geq D^{2-r}$.
This concludes the proof.
\end{proof}

The bound on the expectation value \eqref{eq:entropyexpectationbound} is of the form
\begin{equation}
\mathbb{E} \tr[\rho_A^2]\leq d^{-l}+O(D^{-2})+O(d^{l-n}).
\end{equation}
In comparison, a bound for the expectation value of the form $d^{-l}+O(D^{-1/10})$ for $D\geq n^5$ has been obtained in Ref.~\cite{collins2012matrix} for random translation-invariant MPS.
For periodic boundary conditions, a similar result has been proven in Ref.~\cite{gonzales_spectral_2018}.
Perhaps not surprisingly, our bound in the disordered case scales slightly better in $D$. Moreover, the bond dimension
$D$ is not required to grow with the system size $n$.

In Ref.~\cite{collins2012matrix}, 
Levy's Lemma has been employed in order to obtain an exponential concentration bound. This is not possible in our case, as the Lipschitz constant of the purity is necessarily lower bounded by $O(D)$. As a matter of fact, consider the MPS generated by choosing 
\begin{equation}
U^{(k)}:=\mathbb 1_d\otimes \mathbb 1_D 
\end{equation}
for all $k$, i.e., $\rho=D^2|0\rangle\langle 0|^{\otimes n}$. Then for any $A$, $\tr[\rho_A^2]=D^2$. Now pick a site $k$ and a unitary 
\begin{equation}
U:=\mathbb{1}_d\otimes V
\end{equation}
with 
$\mathrm{Tr}(V)=0$.
 For example, we can choose \begin{equation}
 V:=\mathrm{diag}(e^{2\pi \mathrm{i}j/D})_{j=0,\dots ,D-1}.
\end{equation}
We construct the MPS $\sigma$ with the identity on every site and $U$ on site $k$.
 For any $A$ containing $k$, $\tr[\sigma_A^{2}]=0$. 
 The Lipschitz constant of the function
 \begin{equation}
\begin{aligned}
    U(dD)^{\times n}\to \mathbb R,\\
    (U)\mapsto \tr[{\rho^U_{A}}^2]
    \end{aligned}
\end{equation} 
where $\rho^U$ is the MPS constructed with the unitaries and where we equip  $U(dD)^{\times n}$ with the product Frobenius norm is then bounded by
\begin{equation}
    L\geq \frac{|\tr[\rho_A^{2}]-\tr[\sigma_A^{2}]|}{||\mathbb 1- U||_2}\geq \frac{D^2}{2D}=O(D),
\end{equation}
where we  have made use of 
the triangle inequality to further 
bound the denominator.

\section{Concentration around unit norm}\label{sec:normconcentration}

In this section, we will make use 
the machinery of statistical mechanics mappings in order to show that the norm of the vectors $|\psi\rangle$ is exponentially concentrated around unity, i.e. around being an actual quantum state.
\begingroup
\def\thelemma{\ref{lemma:boundonnorm}}
\begin{lemma}[Concentration 
around a unit norm]
 \begin{equation}
 \mathrm{Pr}\left(|\langle\psi|\psi\rangle-1|\geq \varepsilon\right)\leq \varepsilon^{-2} d^{-n}.
 \end{equation}
 \end{lemma}
 \addtocounter{lemma}{-1}
\endgroup

To prove this statement, first notice that
\begin{equation}
\mathbb{E}\langle\psi|\psi\rangle^2=\tr[\mathbb{F}(|\psi\rangle\langle\psi|)^{\otimes 2}].
\end{equation}
Moreover, using first order Weingarten calculus it is easy to see that $\mathbb{E}\langle\psi|\psi\rangle=1$.
Graphically, this corresponds to the spin chain depicted in Eq.~\eqref{eq:partitionrenyi} with only blue plaquettes as
\begin{multline}
\mathbb{E}\langle\psi|\psi\rangle^2\\
=\sum_{\{\mathbb{1},\mathbb{F}\}^n}\ipic{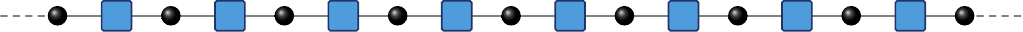}{0.44}.
\end{multline}
However, given the values of the plaquettes in~\eqref{eq:plaquettes}, the only non-zero contributions are all spins $\mathbb{1}$ or all spins $\mathbb{F}$.
This yields
\begin{equation}\label{eq:expectedvaluenorm}
\mathbb{E}\langle\psi|\psi\rangle^2=1+\left(\frac{dD^2-d}{D^2d^2-1}\right)^n\leq 1+ d^{-n}.
\end{equation}
Now we obtain from Markov's inequality
\begin{equation}
\mathrm{Pr}\left((\langle\psi|\psi\rangle-1)^2\geq\varepsilon\right)\leq \varepsilon^{-1}\mathbb{E}(1-\langle\psi|\psi\rangle)^2\leq \varepsilon^{-1} d^{-n}.
\end{equation}
This exponential concentration allows us to prove all other concentration results without regarding the normalization.
In fact, we combine every calculation with a union bound and the above concentration result such that the normalization does not change the statements. 
It is also noteworthy that all these concentration results follow from simple applications of Markov's inequality and do not require geometric methods such as Levy's Lemma ~\cite{ledoux2001concentration}.

\section{Local expectation values}
Consider an observable $O$ acting on $(\mathbb{C}^{d})^{\otimes l}$ for some small number of sites $l$.
Here we consider the expectation values $\langle\psi| O\otimes \mathbb{1}_{(\mathbb{C}^2)^{\otimes n-l}}|\psi\rangle$.
We assume w.l.o.g. that $O$ is traceless.
First, note that we have 
\begin{equation}
\mathbb{E}\langle\psi| O\otimes \mathbb{1}_{(\mathbb{C}^2)^{\otimes n-l}}|\psi\rangle=\tr(O)=0.
\end{equation}
We would like to show that the expectation values are concentrated around $0$. 
One possible way to do this is to exploit results 
of Section~\ref{sec:maximumentropy} 
to argue that the local density matrix is close in norm with high probability to the maximally mixed state, as done in \cite{collins2012matrix}, and conclude by Hölder's inequality that for the expectation value expressed in terms of the reduced density matrix $\rho$ the following holds:
\begin{equation}
\begin{aligned}
    |\tr[O\rho]|&=|\tr[O(\rho-\mathbb{1}/d^{l})]|\leq ||\rho-\mathbb{1}/d^{l}||_{\infty} ||O||_1,
\end{aligned}
\end{equation}
where $||O||_1$ is a constant in $n$ and $D$.
Nevertheless, we would like to showcase how our method can be applied directly to this problem. 

To showcase the flexibility of our approach, we provide a direct bound on the second moment
\begin{equation}
\mathbb{E}\langle\psi| O\otimes \mathbb{1}_{(\mathbb{C}^2)^{\otimes n-l}}|\psi\rangle^2=\mathbb{E}\tr[(O\otimes \mathbb{1}_{(\mathbb{C}^2)^{\otimes n-l}}|\psi\rangle\langle\psi|)^{\otimes 2}].
\end{equation}

We consider the case $l=1$ for simplicity.
Graphically, we have
\begin{multline}
\mathbb{E}\langle\psi| O|\psi\rangle^2\\
=\ipic{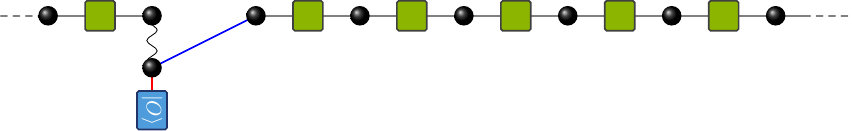}{0.45},
\end{multline}
where $|O\rangle:=[(\mathbb{1}\otimes O)|\Omega\rangle]^{\otimes 2}$.
We obtain the  interactions
\begin{align}
\begin{split}
&\mathbb{1}\;\ipic[-0.7]{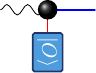}{0.7}\,\mathbb{1}=-\frac{\tr[O^2]}{D^2d^3-d},\\ &\mathbb{F}\;\ipic[-0.7]{partitionfractionO}{0.7}\;\mathbb{F}=\frac{\tr[O^2]D^2}{D^2d^2-1},\\
&\mathbb{1}\;\ipic[-0.7]{partitionfractionO}{0.7}\,\mathbb{F}=-\frac{\tr[O^2]D}{D^2d^3-d},\\ &\mathbb{F}\;\ipic[-0.7]{partitionfractionO}{0.7}\;\mathbb{1}=\frac{\tr[O^2]D}{D^2d^2-1}
\end{split}
\end{align}
from this.
Simply by ignoring the one negative contribution (all spins $\mathbb{1}$), we obtain 
\begin{align}
\begin{split}
&\mathbb{E}\langle\psi| O\otimes \mathbb{1}_{(\mathbb{C}^2)^{\otimes n-l}}|\psi\rangle^2\\
&\textcolor{white}{aaaa}\leq \tr[O^2]\left(\frac{D}{D^2d^2-1}\new{\eta(D,d)}\sum_{i=\new{0}}^{\new{n-2}}\new{\eta(d,D)^i}\right.\\
&\textcolor{white}{aaaa}\textcolor{white}{\leq}\left.+\frac{D^2}{D^2d^2-1}\left(\frac{d^2D-D}{D^2d^2-D}\right)^{n-2}\right)\\
&\textcolor{white}{aaaa}\leq \tr[O^2]\left(D^{-2}\new{\frac{1-d^{-n+1}}{1-d^{-1}}}+D^{-n+1}\right)\\
&\textcolor{white}{aaaa}\leq 2D^{-2}\tr[O^2]
\end{split}
\end{align}
for $n\geq 2$ as a combinatorial bound.
We can now achieve a concentration result for local expectation values from Markov's inequality.
Previously, typicality of expectation values was argued for translation-invariant MPS in Ref.~\cite{garnerone2010typicality}, where the authors argue an exponential concentration if $D$ grows faster than $\Omega(n^2)$. 
However, their proof is based on the concentration of measure phenomenon and the bound they obtain on the Lipshitz constant involves the assumption that the second highest eigenvalue of the transfer operator is small enough for essentially all instances. While a numerical assessment 
indicates that this is the case for all practical purposes, it is difficult to obtain explicit rigorous bounds for the dependence on $n$ this way.
Our simple bound from second moments only scales as $D^{-2}$ (as opposed to $e^{-\Omega(D)}$) but it is independent of $n$ and does not require any further assumptions. 

\section{Outlook}


In this work, we have \new{systematically explored} random matrix product states (RMPS) were the individual tensors have been chosen independently according to the Haar measure. Exploiting a mapping to a 
one-dimensional statistical mechanics model, we have
\new{been in the position to} compute 
expectation values of various quantities 
for such disordered RMPS.
We \new{have} derived concentration results for the effective dimension implying equilibration under the time evolution of generic Hamiltonians, and concentration results for the R\'{e}nyi $2$-entropy and the expectation values of local observables.

While we have put properties of such families of quantum states into the focus of our analysis, it should be clear that by means of the parent Hamiltonian concept 
mentioned above, we could have discussed ensembles of
local Hamiltonians. For translational-invariant RMPS 
\cite{Lancien},
this picture is particularly transparent.

An obvious further problem is to consider quantities that require higher moments such as higher R\'{e}nyi entropies $S_{\alpha}$ with $2\leq \alpha\in\mathbb{Z}$.
This would result in more complicated statistical models with $\alpha!$ many local degrees of freedom.
Another open question concerns higher dimensional systems.
It would be interesting to apply a similar analysis to projected entangled pair states.

For  translational-invariant RMPS, the expected correlation length has been 
proven in Refs.~\cite{gonzales_spectral_2018}.
The family of states defined in this fashion gives rise to
generic representatives of the \emph{trivial phase of
matter} \cite{1010.3732} with unit probability.
If a state has symmetries, i.e., if the state vector
satisfies
    $U_g^{\otimes n} |\psi\rangle = 
    e^{i \theta_g}|\psi\rangle
$
for some real phase $\theta_g$, and where 
$g\mapsto U_g$
is a linear unitary representation of a symmetry group $G$, then this symmetry is (for a suitable phase gauge) 
reflected on the virtual level of the 
MPS as a projective unitary representation of group $G$,
satisfying
\begin{equation}
    V_g V_h = e^{i\omega(g,h)}V_{gh}
,
\end{equation}
for a real phase $(g,h)\mapsto \omega(g,h)$
\cite{PhysRevA.79.042308}.
Different \emph{phases of matter} respecting 
these symmetries in symmetry protected topological
order are now captured by equivalence classes,
called cohomology classes, they again
forming a group, the second cohomology
group $H_2(G,U(1))$ of $G$ over $U(1)$ 
\cite{1010.3732,PhysRevB.81.064439,PhysRevB.83.035107}. 
That said, 
it now makes sense to think of RMPS that respect a physical symmetry and think of \emph{typical symmetry protected topological (SPT)
phases of matter}: Here, the Haar measure is chosen in each of the blocks of a 
direct sum on the virtual level, respecting the 
projective unitary representation of group $G$. In this sense, one can speak of common representatives
of SPT phases, a line of thought that will be elaborated upon
elsewhere.

More broadly put, this work can be seen as a contribution
to a bigger program concerned with 
\new{understanding generic phases of quantum matter by means of}
\emph{random tensor
networks}. Indeed, properties of random tensor
networks can often be easier computed than those of tensor networks in which the entries are specifically chosen:
Randomness hence serves as a computational tool,
a line of thought that can be dated back to 
Ref.\ \cite{hayden2006aspects} and further.

In Ref.\ \cite{Hayden2016}, such a line of thought has already
been applied to identify properties of holographic
tensor networks, where a desirable property to be a
so-called perfect tensor turns out to the 
approximately satisfied with high probability. Building
upon this insight, one can compute properties
of the resulting boundary state. More ambitiously
still, it makes a lot of sense to think of 
holographic random tensor network models in which
the tensors have further structure, e.g., to be
match-gate tensor networks \cite{Jahn}. Similarly,
questions of properties of higher-dimensional cubic
tensor networks arise along similar lines. It is the 
hope that the present work can provide insights and a
\new{powerful} machinery to address such further questions
when exploring typical instances of phases of matter.\\

\subsection*{Acknowledgements}
We would like to warmly thank Alexander Altland, Alex Goe\ss mann, Dominik Hangleiter, Nick Hunter-Jones, Richard Kueng, Amin Thainat, and Carolin Wille for fruitful discussions. 
We would like to thank the DFG 
(CRC 183, project A03,
FOR 2724, EI 519/15-1, EI 519/17-1) and the FQXi for support.

\section{Appendix}
A natural follow up question is whether RMPS have features similar to generic Haar-random states.
Naively, an argument as above might be used to show that RMPSs form approximate spherical $2$-designs \cite{Designs}.
The difference of moment operators in 2-norm can be bounded using the frame potential as
\begin{equation}
\left|\left|\mathbb{E}_{\psi\sim\nu}(|\psi\rangle\langle\psi|)^{\otimes 2}-\frac{2 P_{\mathrm{sym}}}{d^n(d^n+1)}\right|\right|^2_F=\mathcal{F}_{2,\nu}-\frac{2}{d^n(d^n-1)},
\end{equation}
and the frame potential is given by
\begin{equation}
\mathcal{F}_{2,\nu}:=\mathbb{E}_{\psi,\phi\sim\nu}|\langle \psi|\phi\rangle|^4,
\end{equation}
which is reminiscent of the expression in~\eqref{eq:overlapbound}.
After all, random product states constitute an exact projective $1$-design.

However, RMPS with polynomially bounded bond dimensions have low-entanglement structure by definition.
Since the Schmidt rank along any bi-partition is bounded by $D$, the purity is bounded below by $1/D$ but the average over this quantity is exponentially small for an approximate $2$-design.
In more detail, consider a probability measure such that
\begin{equation}
\left|\left|\mathbb{E}_{\psi\sim\nu} (|\psi\rangle\langle\psi|)^{\otimes 2}-\frac{2 P_{\mathrm{sym}}}{d^n(d^n+1)}\right|\right|_{F}\leq \varepsilon\,.
\end{equation}
With this notion we obtain a straightforward bound on the entanglement purity over a bi-partition of the spin chain into subsets $A$ and $B$ of equal size $n/2$.
With the norm inequality $||\cdot||_1\leq \sqrt{\dim\mathcal{H}}||\cdot||_F$ (see, e.g., Ref.~\cite[Eq.~(1.2.6)]{low_pseudo-randomness_2010}), we obtain
\begin{align}
\begin{split}
&\mathbb{E}_{\psi\sim\nu}\left(\tr\left[\rho_A^2\right]\right)\\
&=\mathbb{E}_{\psi\sim\nu}\left(\tr\left[\mathbb{F}_{A,\overline{A}}\rho^{\otimes 2}_A\right]\right)\\
&=\tr\left[\mathbb{F}_{A,\overline{A}}\tr_{B,\overline{B}}\left[\mathbb{E}_{\psi\sim\nu} \left(|\psi\rangle\langle\psi|\right)^{\otimes 2}\right]\right].
\end{split}
\end{align}
Therefore
\begin{align}
\begin{split}
&\mathbb{E}_{\psi\sim\nu}\left(\tr\left[\rho_A^2\right]\right)\\
&\leq\frac{2}{d^n(d^n+1)}\tr\left[\mathbb{F}_{A,\overline{A}}\tr_{B,\overline{B}}\left[P_{\mathrm{sym}}\right]\right]+d^n\varepsilon\\
&=\frac{1}{d^n(d^n+1)}\tr\left[\mathbb{F}_{A,\overline{A}}\tr_{B,\overline{B}}\left[\mathbb{F}_{AB,\overline{AB}}+\mathbb{1}_{AB,\overline{AB}}\right]\right]+d^n\varepsilon\\
&=\frac{2 d^{3n/2}}{d^n(d^n+1)}+d^n\varepsilon\,.
\end{split}
\end{align}
As a small expectation value implies the existence of instances with small values, this leads to a contradiction for 
\begin{equation}
    \varepsilon\leq \frac{d^{-n}}{D}.
\end{equation}
This argument rules out a ``cut-off'' phenomenon, where the error 
is exponentially suppressed
in the bond dimension $D$ only after some polynomial threshold has been
surpassed.

\bibliographystyle{plain}
\bibliography{Bib/BigReferences37}

\begin{thebibliography}{10}

\bibitem{RevModPhys.91.021001}
D.~A. Abanin, E.~Altman, I.~Bloch, and M.~Serbyn.
\newblock Colloquium: Many-body localization, thermalization, and entanglement.
\newblock {\em Rev. Mod. Phys.}, 91:021001, 2019.

\bibitem{RMT}
G.~Akemann, J.~Baik, and P.~Di Francesco, editors.
\newblock {\em The Oxford handbook of random matrix theory}.
\newblock Oxford University Press, 2015.

\bibitem{GoogleSupremacy}
F.~Arute et~al.
\newblock Quantum supremacy using a programmable superconducting processor.
\newblock {\em Nature}, 574:505--510, 2019.

\bibitem{Bauer}
B.~Bauer and C.~Nayak.
\newblock Area laws in a many-body localised state and its implications for
  topological order.
\newblock {\em J. Stat. Mech.}, P09005, 2013.

\bibitem{brandao_local_2016}
F.~G. S.~L. Brand\~ao, A.~W. Harrow, and M.~Horodecki.
\newblock Local {random} {quantum} {circuits} are {approximate}
  {polynomial}-{designs}.
\newblock {\em Commun. Math. Phys.}, 346:397--434, 2016.

\bibitem{ComplexityGrowth}
F.~G. S.~L. Brandao, W.~Chemissany, N.~Hunter-Jones, R.~Kueng, and J.~Preskill.
\newblock Models of quantum complexity growth.
\newblock arXiv:1912.04297.

\bibitem{brouwer1996diagrammatic}
P.~W. Brouwer and C.~W.~J. Beenakker.
\newblock Diagrammatic method of integration over the unitary group, with
  applications to quantum transport in mesoscopic systems.
\newblock {\em J. Math. Phys.}, 37:4904--4934, 1996.

\bibitem{chandran2015semiclassical}
A.~Chandran and C.~R. Laumann.
\newblock Semiclassical limit for the many-body localization transition.
\newblock {\em Phys. Rev. B}, 92(2):024301, 2015.

\bibitem{PhysRevB.83.035107}
X.~Chen, Z.-C. Gu, and X.-G. Wen.
\newblock Classification of gapped symmetric phases in one-dimensional spin
  systems.
\newblock {\em Phys. Rev. B}, 83:035107, 2011.

\bibitem{collins2012matrix}
B.~Collins, C.~E. Gonz{\'a}lez-Guill{\'e}n, and D.~P{\'e}rez-Garc{\'\i}a.
\newblock Matrix product states, random matrix theory and the principle of
  maximum entropy.
\newblock {\em Commun. Math. Phys.}, 320:663, 2013.

\bibitem{collins2006integration}
B.~Collins and P.~{\'S}niady.
\newblock {Integration with respect to the Haar measure on unitary, orthogonal
  and symplectic group}.
\newblock {\em Comm. Math. Phys.}, 264:773--795, 2006.

\bibitem{AreaReview}
J.~Eisert, M.~Cramer, and M.~B. Plenio.
\newblock Area laws for the entanglement entropy.
\newblock {\em Rev. Mod. Phys.}, 82:277, 2010.

\bibitem{1408.5148}
J.~Eisert, M.~Friesdorf, and C.~Gogolin.
\newblock Quantum many-body systems out of equilibrium.
\newblock {\em Nature Phys.}, 11:124--130, 2015.

\bibitem{raey}
M.~Fannes, B.~Nachtergaele, and R.F. Werner.
\newblock Finitely correlated states on quantum spin chains.
\newblock {\em Commun. Math. Phys.}, 144:443--490, 1992.

\bibitem{1409.1252}
M.~Friesdorf, A.~H. Werner, W.~Brown, V.~B. Scholz, and J.~Eisert.
\newblock Many-body localisation implies that eigenvectors are matrix-product
  states.
\newblock {\em Phys. Rev. Lett.}, 114:170505, 2015.

\bibitem{garnerone2010statistical}
S.~Garnerone, T.~R. de~Oliveira, S.~Haas, and P.~Zanardi.
\newblock Statistical properties of random matrix product states.
\newblock {\em Phys. Rev. A}, 82:052312, 2010.

\bibitem{garnerone2010typicality}
S.~Garnerone, T.~R. de~Oliveira, and P.~Zanardi.
\newblock Typicality in random matrix product states.
\newblock {\em Phys. Rev. A}, 81:032336, 2010.

\bibitem{christian_review}
C.~Gogolin and J.~Eisert.
\newblock Equilibration, thermalisation, and the emergence of statistical
  mechanics in closed quantum systems.
\newblock {\em Rep. Prog. Phys.}, 79:56001, 2016.

\bibitem{gonzales_spectral_2018}
C.~E. Gonzales-Guillen, M.~Junge, and I.~Nechita.
\newblock On the spectral gap of random quantum channels.
\newblock arXiv:1811.08847.

\bibitem{Designs}
D.~Gross, K.~Audenaert, and J.~Eisert.
\newblock {\em J. Math. Phys.}, 48:052104, 2007.

\bibitem{Webs}
D.~Gross and J.~Eisert.
\newblock Quantum computational webs.
\newblock {\em Phys. Rev. A}, 82:040303({R}), 2010.

\bibitem{gross2009most}
D.~Gross, S.~T. Flammia, and J.~Eisert.
\newblock Most quantum states are too entangled to be useful as computational
  resources.
\newblock {\em Phys. Rev. Lett.}, 102:190501, 2009.

\bibitem{Homeopathy}
J.~Haferkamp, F.~Montealegre-Mora, M.~Heinrich, J.~Eisert, D.~Gross, and
  I.~Roth.
\newblock Quantum homeopathy works: Efficient unitary designs with a
  system-size independent number of non-clifford gates.
\newblock arXiv:2002.09524.

\bibitem{hayden2004randomizing}
P.~Hayden, D.~Leung, P.~W. Shor, and A.~Winter.
\newblock Randomizing quantum states: Constructions and applications.
\newblock {\em Comm. Math. Phys.}, 250:371--391, 2004.

\bibitem{hayden2006aspects}
P.~Hayden, D.~W. Leung, and A.~Winter.
\newblock Aspects of generic entanglement.
\newblock {\em Comm. Math. Phys.}, 265:95--117, 2006.

\bibitem{Hayden2016}
P.~Hayden, S.~Nezami, X.-L. Qi, N.~Thomas, M.~Walter, and Z.~Yang.
\newblock Holographic duality from random tensor networks.
\newblock {\em J. High En. Phys.}, 2016:9, 2016.

\bibitem{huang_instability_2019}
Y.~Huang and A.~W. Harrow.
\newblock Instability of localization in translation-invariant systems.
\newblock 2019.
\newblock arXiv:1907.13392.

\bibitem{hunter-jones_unitary_2019}
N.~Hunter-Jones.
\newblock Unitary designs from statistical mechanics in random quantum
  circuits.
\newblock {\em arXiv:1905.12053}, 2019.

\bibitem{huse2014phenomenology}
D.~A. Huse, R.~Nandkishore, and V.~Oganesyan.
\newblock Phenomenology of fully many-body-localized systems.
\newblock {\em Phys. Rev. B}, 90:174202, 2014.

\bibitem{Jahn}
A.~Jahn, M.~Gluza, F.~Pastawski, and J.~Eisert.
\newblock Holography and criticality in matchgate tensor networks.
\newblock {\em Science Advances}, 5:eaaw0092, 2019.

\bibitem{kliesch_guaranteed_2019}
M.~Kliesch, R.~Kueng, J.~Eisert, and D.~Gross.
\newblock Guaranteed recovery of quantum processes from few measurements.
\newblock {\em Quantum}, 3:171, 2019.

\bibitem{Lancien}
C.~Lancien and D.~Pérez-García.
\newblock {Correlation length in random MPS and PEPS}.
\newblock {\em arXiv:1906.11682}.

\bibitem{ledoux2001concentration}
M.~Ledoux.
\newblock {\em The concentration of measure phenomenon}.
\newblock Number~89. American Mathematical Soc., 2001.

\bibitem{Linden_etal09}
N.~Linden, S.~Popescu, A.~J. Short, and A.~Winter.
\newblock Quantum mechanical evolution towards thermal equilibrium.
\newblock {\em Phys. Rev. E}, 79:061103, Jun 2009.

\bibitem{low_pseudo-randomness_2010}
R.~A. Low.
\newblock Pseudo-randomness and {learning} in {quantum} {computation}.
\newblock {\em arXiv:1006.5227}, 2010.

\bibitem{movassagh2019ergodic}
R.~Movassagh and J.~Schenker.
\newblock An ergodic theorem for homogeneously distributed quantum channels
  with applications to matrix product states.
\newblock {\em arXiv preprint arXiv:1909.11769}, 2019.

\bibitem{movassagh2020theory}
R.~Movassagh and J.~Schenker.
\newblock Theory of ergodic quantum processes.
\newblock {\em arXiv preprint arXiv:2004.14397}, 2020.

\bibitem{nahum2018operator}
A.~Nahum, S.~Vijay, and J.~Haah.
\newblock Operator spreading in random unitary circuits.
\newblock {\em Phys. Rev. X}, 8:021014, 2018.

\bibitem{neill_blueprint_2017}
C.~Neill, P.~Roushan, K.~Kechedzhi, S.~Boixo, S.~V. Isakov, V.~Smelyanskiy,
  R.~Barends, B.~Burkett, Y.~Chen, and Z.~Chen.
\newblock A blueprint for demonstrating quantum supremacy with superconducting
  qubits.
\newblock {\em Science}, 360:195--199, 2017.

\bibitem{MPSSurvey}
D.~Perez-Garcia, F.~Verstraete, M.~M. Wolf, and J.~I. Cirac.
\newblock Matrix product state representations.
\newblock {\em Quant. Inf. Comp.}, 5$\&$6:401, 2007.

\bibitem{quant-ph/0608197}
D.~Perez-Garcia, F.~Verstraete, M.M. Wolf, and J.~I. Cirac.
\newblock Matrix product state representations.
\newblock {\em Quantum Inf. Comput.}, 7:401, 2007.

\bibitem{PolkovnikovReview}
A.~Polkovnikov, K.~Sengupta, A.~Silva, and M.~Vengalattore.
\newblock Non-equilibrium dynamics of closed interacting quantum systems.
\newblock {\em Rev. Mod. Phys.}, 83:863, 2011.

\bibitem{PhysRevB.81.064439}
F.~Pollmann, A.~M. Turner, E.~Berg, and M.~Oshikawa.
\newblock Entanglement spectrum of a topological phase in one dimension.
\newblock {\em Phys. Rev. B}, 81:064439, 2010.

\bibitem{reimann_foundation_2008}
P.~Reimann.
\newblock Foundation of statistical mechanics under experimentally realistic
  conditions.
\newblock {\em Phys. Rev. Lett.}, 101:190403, 2008.

\bibitem{rolandi2020extensive}
A.~Rolandi and H.~Wilming.
\newblock Extensive {R\'{e}nyi} entropies in matrix product states.
\newblock {\em arXiv:2008.11764}, 2020.

\bibitem{sanchez1995simple}
J.~S{\'a}nchez-Ruiz.
\newblock Simple proof of page’s conjecture on the average entropy of a
  subsystem.
\newblock {\em Phys. Rev. E}, 52:5653, 1995.

\bibitem{PhysRevA.79.042308}
M.~Sanz, M.~M. Wolf, D.~P\'erez-Garc\'{\i}a, and J.~I. Cirac.
\newblock {Matrix product states: Symmetries and two-body Hamiltonians}.
\newblock {\em Phys. Rev. A}, 79:042308, 2009.

\bibitem{1010.3732}
N.~Schuch, D.~Perez-Garcia, and I.~Cirac.
\newblock Classifying quantum phases using matrix product states and projected
  entangled pair states.
\newblock {\em Phys. Rev. B}, 84:165139, 2011.

\bibitem{PhysRevB.98.134204}
C.~S\"underhauf, D.~P\'erez-Garc\'{\i}a, D.~A. Huse, N.~Schuch, and J.~I.
  Cirac.
\newblock Localization with random time-periodic quantum circuits.
\newblock {\em Phys. Rev. B}, 98:134204, 2018.

\bibitem{tasaki_quantum_1998}
H.~Tasaki.
\newblock From quantum dynamics to the canonical distribution: general picture
  and a rigorous example.
\newblock {\em Phys. Rev. Lett.}, 80:1373, 1998.

\bibitem{wilming_entanglement_2019}
H.~Wilming, M.~Goihl, I.~Roth, and J.~Eisert.
\newblock Entanglement-ergodic quantum systems equilibrate exponentially well.
\newblock {\em Phys. Rev. Lett.}, 123:200604, 2019.

\bibitem{zhou2019emergent}
T.~Zhou and A.~Nahum.
\newblock Emergent statistical mechanics of entanglement in random unitary
  circuits.
\newblock {\em Phys. Rev. B}, 99:174205, 2019.

\end{thebibliography}

\end{document}